\documentclass[a4,paper,11pt]{article}

\usepackage{cite}
\usepackage{amsmath,amssymb,amsfonts}
\usepackage{amsthm}
\usepackage{algorithmicx,xpatch}
\usepackage{algorithm}
\usepackage{algpseudocode}
\usepackage{graphicx}
\usepackage{textcomp}
\usepackage{xcolor}

\usepackage[utf8]{inputenc}
\usepackage[T1]{fontenc}
\usepackage{stmaryrd}
\usepackage{listings}
\usepackage{tikz}
\usepackage{setspace}
\usepackage{multirow}
\usepackage{url}
\usepackage{moresize}

\usepackage{enumitem}
\usepackage{pifont}
\usepackage{booktabs}
\usepackage{etoolbox}
\usepackage{verbatim}
\usepackage{float}
\usepackage{ifthen}
\usepackage{scalefnt}
\usepackage{blkarray}
\usepackage{subfig}
\usepackage[mathscr]{eucal}

\newcommand{\nzero}{\setminus\{0\}}

\newcommand{\N}{\mathbb{N}}

\newcommand{\zp}{\mathbb{Z}/p\mathbb{Z}}

\newcommand{\zx}{\mathbb{Z}[X]}
\newcommand{\znx}{\mathbb{Z}_{n-1}[X]}
\newcommand{\znnx}{\mathbb{Z}_{n}[X]}

\newcommand{\iphi}{\phi^{-1}}

\newcommand{\mnsB}{\mathcal{B}}

\newcommand{\isrep}{\underset{\mnsB}{\equiv}}

\DeclareMathOperator{\mulwlo}{mul\_w\_lo}
\DeclareMathOperator{\mulwhi}{mul\_w\_hi}
\DeclareMathOperator{\UP}{UP}

\theoremstyle{definition}
\newtheorem{defi}{Definition}
\theoremstyle{remark}
\newtheorem{remark}{Remark}
\newtheorem{theorem}{\textbf{Theorem}}
\usepackage{cite}
\usepackage{amsmath,amssymb,amsfonts}
\usepackage{graphicx}
\usepackage{textcomp}
\usepackage{xcolor}
\usepackage{hyperref}
\usepackage{array}
\usepackage{diagbox}
\usepackage{multirow}
\usepackage{booktabs}

\def\BibTeX{{\rm B\kern-.05em{\sc i\kern-.025em b}\kern-.08em
    T\kern-.1667em\lower.7ex\hbox{E}\kern-.125emX}}

\usepackage{authblk}

\author[1]{DIDIER Laurent-Stéphane}
\author[1]{DUTOIS-RUIZ Alexy}
\author[1]{ROBERT Jean-Marc}

\affil[1]{IMATH - Université de Toulon, France}

\title{Faster PMNS Multi-precision Multiplications Using Truncated Montgomery Technique
\thanks{This work has been funded by the MAERA ANR Project}
}

\begin{document}
\maketitle

\begin{abstract}

The Polynomial Modular Number Systems (PMNS) aim to represent elements of fields or rings of large characteristics using polynomials satisfying bounds on some parameters (degree, absolute values of the coefficients). Those PMNS, while some conditions are fulfilled, using convenient parameters and implementation features, allow some speed-ups in cryptographic computations. Recent works (Meloni \emph{et al.} \cite{MeloniPV25}) propose better speed-ups by improving the parameter generation of the system, and by using multi-precision polynomial coefficients, i.e. coefficients stored using several machine words. In this work, we first present PMNS schoolbook software implementation better in some context than the Toeplitz state-of-the-art counterpart of \cite{MeloniPV25}, taking advantage of smaller memory cost by being free to choose a smaller $n$ parameter, minimizing at the same time the complexity of the multiplication. This implementation for 4096 bit modulo size shows element size smaller by 11\% and 20 \% speed-up. We thus present a new improvement on the PMNS, applying to the \textsl{internal reduction} an approach similar to the truncated Montgomery reduction technique presented by Didier \emph{et al.} in \cite{DidierEGR24}. In the context of software implementations using \texttt{AVX512} instruction set extension, and modulo size up to 8192 bits, this new approach allows speed-ups up to 15\% in modular multiplication computation, in comparison with conventional approaches.% or Karatsuba polynomial multiplication. We also present in the same context, the first squaring, which is about 15 \% faster than its multiplication counterpart.
\end{abstract}

\textbf{keywords:} PMNS, Multi-Precision, Truncated Montgomery Modular multiplication, Internal Reduction, SIMD
%\end{keywords}

\section{Introduction}

The Polynomial Modular Number System (PMNS) is a number representation by polynomials that is an interesting alternative to classical multi-precision modular arithmetic \cite{MeloniPV25} \cite{9974256} because of its parallel properties. This makes it attractive for cryptographic purposes \cite{sike} .

However, this system becomes less useful for large numbers because the required polynomials have a degree that is too high, exceeding the number of words required to store a multi-precision integer \cite{9974256}. Some improvements have been suggested to address this issue, such as the use of multi-precision coefficients \cite{MeloniPV25}. 

In this article, we show how to reduce the cost due to the use of multi-precision coefficients by using truncation technique on each coefficients similar to \cite{DidierEGR24}.
Throughout this paper, we will more particularly compare our proposition to \cite{MeloniPV25} which, to the best of our knowledge,  is currently the state of the art.

\subsection*{Contributions}

Firstly, we present a \textit{truncated} Montgomery-like internal reduction algorithm inspired by \cite{DidierEGR24} in order to reduce the complexity of this algorithm.
%{\color{red} est-ce gain de ce papier ou de \cite{DidierEGR24} ?}
Thus, we give the theoretical amount of elementary multiplications of the complete and of the truncated algorithms.

We then provide software implementations of both configurations (full and truncated) and we confirm the reduction of elementary operation numbers as well as the actual number of multiplication-accumulation instructions in the vectorized implementation using \texttt{AVX512} and \texttt{IFMA52} instruction set.

We follow comparing benchmark results from our code with the software implementation from \cite{MeloniPV25}, which is the state-of-the-art, to the best of our knowledge, on Sapphire Rapids and ZEN 4 architectures. This experimentation shows that our full implementation is competitive with it, and that applying our truncated Montgomery-like internal reduction leads to a speedup up to 15\% for the multiplication in the PMNS compared to our full implementation.

% NOTE peut-être pas intéressant
%We then enumerate the number of costly instructions in the three multiplication algorithms.
%By doing so, we show how a classical schoolbook multiplication can exhibit better performances in the PMNS compared to a subquadratic efficient Toeplitz vector-matrix multiplication.

%{\color{red}Je penses, qui'l faille retoucher l'introduction. Il me semble que son objectif pourrait être plus clair }

\subsection*{Notations}

In this article, $w$ is the size of a machine word and $s$ the number of machine words required to write a multi-precision number.
$a[i]$ refers to the $i^{th}$ word of $a$ and $a \langle i \rangle$ to the $i^{th}$ bit of $a$.
$\bigvee$ is the logical OR and $>>$ the bitwise right shift.

%\begin{itemize}
%    \item First vectorized squaring in PMNS, with small and large coefficients
%    \item We present PMNS schoolbook software implementation better than the Toeplitz state-of-the-art counterpart of \cite{MeloniPV25}, taking advantage of smaller memory cost by being free to choose the $n$ parameter, minimizing at the same time the complexity of the multiplication 
%    \item Starting from the previous software implementations, we adapted the Truncated Montgomery approach of Didier \emph{et al.}\cite{DidierEGR24} to the PMNS with large coefficients, leading to significant speed-ups.
%\end{itemize}

%The speed-ups are ...

%Mionkowski bound \cite{BMPV2024}

\section{Background on PMNS}
Polynomial Modular Number Systems (PMNS) are non-positional, modular integer number systems that have been pioneered by Bajard et al.\cite{amns_04}. In such number systems integers modulo $p$ are represented as polynomials with integer coefficients. 

A PMNS is defined by a tuple $(p, n, \gamma, \rho, E)$, where $p$, $n$, $\gamma$ and $\rho$ are nonzero positive integers, and $E\in\mathbb{Z}[X]$ is a monic polynomial, such that $\deg(E)=n$ and $E(\gamma) \equiv 0 \bmod p$. Thus, an integer $a$ is represented in this system by a polynomial $A[X]$ such that $\deg(A) < n$ and with coefficients bounded by $\rho$ such that $A(\gamma) \equiv a \bmod p$. More formally a PMNS is defined as follows:

\begin{defi} \label{def:pmns}
Let $p \geqslant 3$ and $n, \rho \geqslant 2$ be three integers.\\
Let $E\in \znnx$ be a monic polynomial and $\gamma \in \zp\nzero$, such that $E(\gamma) \equiv 0 \bmod{p}$. A tuple $\mnsB = (p, n, \gamma, \rho, E)$  is a PMNS if
$\forall a \in \zp$, $\exists A \in \znnx$ such that:
\begin{enumerate}
\item $A(\gamma) \equiv a \mod p$
\item $ \Vert A \Vert _{\infty} < \rho $
\end{enumerate}
\end{defi}

Such a polynomial $A$ is called a \textbf{representation} of $a$ in $\mnsB$ and we denote $A \isrep a$.
We will denote a polynomial belonging to this PMNS by $A \in \mnsB$.

As a consequence, the arithmetic operations in PMNS are polynomial operations. However, they must be performed in order to ensure outputs in $\mnsB$. Indeed, if $A, B \in \mnsB$, then $\|A+B\|_{\infty}$ and $\|A\times B\|_{\infty}$ might be greater than $\rho$. Furthermore, $\deg(A\times B)$ is most-likely to be greater than $n$. In the first case,  an \textsl{internal reduction} has to be performed and an \textsl{external reduction} for the latter.

\subsection{External reduction} \label{subsubsec:ext-red}
Let $R\in \zx$. The \textit{external reduction} consists in computing a polynomial $C \in \znx$ such that $C(\gamma) \equiv R(\gamma) \bmod p$, using $E$ of $\mnsB$, as follows: $$C \equiv R \bmod E\,.$$
Since $E\in\znnx$ and is monic, $C \in \znx$. Moreover, $E(\gamma) \equiv 0 \bmod p$ ensures that $C(\gamma) \equiv R(\gamma) \bmod p$. 
The polynomial $E$ is called the {\em external reduction polynomial}.

Some criteria have been proposed  in order to find a \textsl{good} {\em external reduction polynomial} \cite{BMPV2024}. Among them,  irreducible sparse polynomials are the best candidates. In \cite{dosso2025friendly}, it is shown how to adapt the PMNS parameters generation in order to deal with non-monic polynomials. This result extends the set of possible {\em external reduction polynomial}. 

The external reduction can also be seen as a vector-matrix multiplication where the matrix is a Toeplitz matrix \cite{MeloniPV24, dosso2025friendly, MeloniPV25}. The advantage of this approach relies on the sub-quadratic complexity of this product \cite{hasan2012multiway}.

\subsection{Internal reduction} \label{subsubsec:int-red}

Let $C\in \znx$. The \textit{internal reduction} aims to compute a polynomial $S \in \mnsB$ such that $S(\gamma) \equiv C(\gamma) \bmod{p}$. In this paper, we focus on the Montgomery-like approach presented in~\cite{amns_08}, see Algorithm~\ref{alg:redcoeff}. It is one of the most efficient internal reduction methods \cite{th_dosso} and requires three additional parameters, two polynomials $M, M'$ and an integer $\phi \geqslant 2$ such that:
\begin{center}
   $M \isrep 0$ and $M' \equiv -M^{-1}\ \text{mod}(E, \phi)$
\end{center}
\begin{algorithm}
  \caption{Coefficients reduction (RedCoeff)~\cite{amns_08}}
  \label{alg:redcoeff}
  \begin{algorithmic}[1]
    \Require  $\mathcal{B} = (p, \, n, \, \gamma, \, \rho, E)$ a PMNS, $C \in \mathbb{Z}_{n-1}[X]$ such that $\Vert C \Vert_{\infty} < 2 \rho$, $\phi \in \mathbb{N}\setminus\{0\}$, $M \isrep 0$,
    and $M' \equiv -M^{-1}\ \text{mod}(E, \phi)$.
    \Ensure $S(\gamma) \equiv C(\gamma)\phi^{-1} \bmod p$, and $S \in \mnsB$
    \State $Q \leftarrow C \times M' \bmod (E, \, \phi)$
    \State $T \leftarrow Q \times M \bmod E$
	\State $S \leftarrow (C + T) / \phi$
    \State return $S$
  \end{algorithmic}
\end{algorithm}

%\begin{remark}
This method is efficient if $\phi$ is a power of two. Also, it induces a multiplicative factor $\iphi$ on the output. In~\cite{DDV19}, it is explained how to deal with this factor. A process to generate efficient PMNS is also presented, where $\phi$ can be any power of two.
%\end{remark}

Algorithm~\ref{alg:pmns_mult} presents the Montgomery-like modular multiplication where RedCoeff refers to Algorithm~\ref{alg:redcoeff}.
\begin{algorithm}
  \caption{Multiplication in PMNS}
  \label{alg:pmns_mult}
  \begin{algorithmic}[1]
    \Require $A \isrep a$, $ B \isrep b$, $\phi \in \mathbb{N}\setminus\{0\}$ and $\mathcal{B} = (p, \, n, \, \gamma, \, \rho, E)$
    \Ensure $S \in \mathcal{B}$ with $S \isrep ab\phi^{-1}$
    \State $C \leftarrow A \cdot B \mod E$
    \State $S \leftarrow \text{RedCoeff}(C)$
    \State return $S$
  \end{algorithmic}
\end{algorithm}

Similarly to the external reduction, it is possible to express the internal reduction as a matrix-vector product of a Toeplitz matrix \cite{MeloniPV24, dosso2025friendly, MeloniPV25}. It appears that some specific primes $p$ lead to PMNS having efficient internal reduction \cite{coladon2021mphell, bouvier2021alternative}. These moduli are generalized as PMNS-friendly primes in \cite{dosso2025friendly}.

\subsection{Complexities}% \color{red} A SUPPRIMMER ? \color{black}}

\subsubsection{Parameters generation} \label{subsubsec:par-gen}
\label{sec:pmnsparamgen}

Building a proper PMNS is a hard task because of the relation between its parameters $p$, $\rho$, $n$ and $\phi$. Some relations between them have been established and a parameter generation process have been proposed in \cite{MeloniPV24, dosso2025friendly}.

PMNS can be seen as $n$-dimensional euclidean lattices that are built from polynomials $P$ such that $P(\gamma) \equiv 0 \bmod p$ \cite{BMPV2024}. As a consequence, because of the Minkowski bound related to lattices \cite{minkowski1910geometrie} the following bound holds \cite{BMPV2024}:
$$\rho > \frac{1}{2}p^{1/n}(1+n-1)\kappa$$
where $\kappa$ is a value related to the largest coefficient of the \textsl{canonical} basis of the lattice. The point here, is that the larger we want the moduli $p$, the larger is the size $\rho$ of the PMNS coefficients or the larger grows $n$. Thus, in order to build PMNS for large moduli $p$ we need to find a trade-off between the degree $n$ of the PMNS polynomials and the use of multi-precision integer arithmetic for the polynomials coefficients.

\subsubsection{PMNS complexity}$ $\\ 
In this section, we discuss PMNS complexities for modular multiplication with Algorithm~\ref{alg:pmns_mult}.
We consider a $w$-bit processor architecture and basic arithmetic computations performed on $w$-bit words. Let $\mathscr M$ and $\mathscr A$ respectively denote the multiplication and the addition of two $w$-bit integers.  Let ${\mathscr L}$ be the notation for logical operations, for instance the AND function, or the right and left shifts on $w$-bit integers.  

The performances and the required memory storage of a PMNS mainly depend on the target architecture and the value of $n$. We assume the modular multiplication inputs belong to a PMNS $\mathcal{B} = (p, n, \gamma, \rho, E)$, such that: $\rho = 2^t$ and $\phi = 2^h$, where $t, h \in \N$ and $1 \leqslant t < h \leqslant w$. 
Since elements in $\mnsB$ are polynomials, $n$ $w$-bit data words are required to represent each of them. Thus, addition is a simple polynomial addition, whose cost is $n \mathscr A$.
%As a consequence, an element in $\mathcal{B}$ requires $n\times w$-bits to be represented.

In this paper, we consider sparse external reduction polynomials $E(X) = X^n - 2$ exclusively. %or $E(X) = X^n - X - 1$.

In \cite{MeloniPV25}, the authors underline the interest of having multi-precision polynomial coefficient, i.e. $\phi = 2^{sw}$ with $s\in \mathbb N$ and $s>1$. They show that with $\phi = 2^w$, ie single machine word coefficients, the value of the parameter $n$ is far greater than the value they call $n_{opt}$ which is the equivalent classic multi-precision representation word number. In comparison, with $s>1$ one has a better compromise for the number representation memory cost, that is $n_m s<n$, denoting $n_m$ the actual value of the PMNS external reduction polynomial degree in case $s>1$.  Thus, as noticed in section \ref{sec:pmnsparamgen}, the elementary multiplication cost is better. 
Table~\ref{tab:ops-cost-pmns} gives the modular multiplication theoretical cost in both cases $s=1$ and $s>1$.
It is worth noticing that our approach uses SIMD \texttt{AVX512} instruction set extension. Therefore the instruction count corresponds to the costs presented in table~\ref{tab:ops-cost-pmns} roughly divided by eight, in case eight multiplications are performed in one single vector multiplier.

\begin{table}[h]
\renewcommand{\arraystretch}{1.2}
\small
\begin{center}
\begin{tabular}{|l|l|}
\hline
$E(X)$ & $X^n - 2$  \\
\hline 
$s=1$ & $3n^2 \mathscr M + (5n^2 - 3n) \mathscr A+ (2n-1) \mathscr L$\\% \mathscr S^w_r + (n - 1)\mathscr S^1_l$ \\
\hline
$s>1$ & $2.5(n_m s)^2 \mathscr M + (5(n_m s)^2 - n_m - 2) \mathscr A + n_m s \mathscr L$ \\
\hline
\end{tabular}
\end{center}
\caption{\label{tab:ops-cost-pmns} PMNS Modular multiplication theoretical cost in elementary word multiplications $\mathscr M$, additions $\mathscr A$ and other logic operations $\mathscr L$ (mainly shifts and maskings), where $\phi = 2^{sw}$, where $s=1$ and $s>1$.}
\end{table}

Recent works \cite{MeloniPV24,MeloniPV25} introduced sub-quadratic approaches in order to lower the complexity of the polynomial multiplication or, in a vector-matrix setting, a Toeplitz vector-matrix product, which is equivalent of a Karatsuba multiplication. These approaches can be recursive. In this case, the number of elementary multiplications  is lower, at the cost of more elementary additions. % due to the final reconstruction, operand split combinations, etc.  
Thus a compromise is often necessary depending on the operand size, limiting the recursive calls to larger sizes, while the last calls for the smallest sizes has to remain with quadratic complexity using a schoolbook approach.

\section{Polynomial Multiprecision Coefficient}

The complexity of the PMNS relies upon the $n$ parameter, which is larger than the word number required to store a multi-precision integer. As mentioned above, we expect that having larger polynomial coefficients would lead to a better instruction count in comparison with single word coefficients.

In practice, the same algorithms are applied for the same operations. For instance, performing a modular multiplication in a PMNS with multi-precision coefficients is done with Algorithm \ref{alg:pmns_mult}. The only one difference relies on the polynomial coefficient size. 

In the sequel, we will consider a PMNS $(p,n,\gamma,\rho,E,\phi)$ with $E = X^n-\lambda$. In a PMNS (see Definition \ref{def:pmns}), one has $ \Vert A \Vert _{\infty} < \rho $. Multi-precision coefficient means that one has the parameter $\rho> 2^w$ and consequently that $\phi = 2^{sw}$ with $w$ the machine word size and $s>1$, as mentioned above.

\begin{remark}
We will use a polynomial description, however, a matrix-vector writing, as in \cite{DossoRV22,MeloniPV24} of all these expressions remain equivalent. Since we aim to describe our idea with the most simple concepts, and without loss of generalities, we chose this way of representation.
\end{remark}

Let $A$ and $B$ be two polynomials in the PMNS:
$$A = \sum_{k=0}^{n-1} a_k X^k  \mbox{~and~} B = \sum_{k=0}^{n-1} b_k X^k$$

We set:

$$a_k = \sum_{\mu=0}^{s-1} a_{k}[\mu]2^{w\mu} \mbox{~and~} b_k = \sum_{\mu=0}^{s-1} b_k[\mu]2^{w\mu}$$

where $a_k[\mu]$ and $b_k[\mu]$ are the $\mu^{th}$ word of respectively the multi-precision coefficients $a_k$ and $b_k$.
Thus, considering step 1 of Algorithm \ref{alg:pmns_mult}, we have:

$$C\gets A\times B= \sum_{k=0}^{n-1} c_k X^k$$

with

$$c_k = \sum_{i+j=k} a_i b_j + \lambda\times\sum_{i+j=k+n} a_ib_j$$

and, with a schoolbook type expression

$$a_ib_j = \sum_{\substack{\ell=0\\\mu+\nu=\ell}}^{\ell=2s-2} a_i[\mu]b_j[\nu]2^{\ell w}$$

\section{Polynomial Multiprecision Coefficient using Montgomery Truncation Technique}

Similarly to \cite{DidierEGR24}, we apply a Montgomery truncated reduction technique to the RedCoeff Algorithm \ref{alg:redcoeff}. The main idea in a classic multiprecision setting is to avoid computations in the Montgomery operation.

Let's remind the Montgomery reduction of a large integer $c$ modulo $p$, with $p<c<p^2$, $p$ a large integer:
one set $p'=(-p)^{-1} \bmod N$, with the Montgomery constant such that $N$ is a power of $2^w$ and $N>p$. The Montgomery reduction computes
\begin{algorithmic}[1]
  \State $q\gets cp'\bmod N$
  \State $r\gets  \frac{c+q\times p} N$
\end{algorithmic}

The idea in \cite{DidierEGR24} is to take advantage of the knowledge that the division by $N$ step 2 is exact, and to evaluate only the most significant words of $q\times p$. The resulting algorithm as presented by the authors in \cite{DidierEGR24} is the Algorithm \ref{alg:trunc_montred}. This approach leads to computation speed-ups up to 20 \% in the context of Montgomery modular reduction of classic multi-precision integers.

\begin{algorithm}[ht]
  \caption{Montgomery modular reduction with truncated multiplication \cite{DidierEGR24}
}
  \label{alg:trunc_montred}
  \small
  \begin{algorithmic}[1]
    \Require  $C < 4N^2$, $N$ the $sw$-bit modulus, $R = 2^{sw}$, precomputed value $N' = (-N)^{-1} \bmod R$.
    \Ensure $S \equiv C \times R^{-1}\bmod N$ and $S<2N$
    \State $q\gets C\times N' \bmod R$\Comment{$q<R$}
    \State $\widehat{qN} \gets \lfloor (q\times N)/R \rfloor$
    \State $c_{add} = \mbox{$\bigvee$}_{i=0}^{sw -1} C \langle i \rangle$ \Comment{OR of all bits of $C$}
    \State $S \gets\lfloor C/R  \rfloor+ \widehat{qN} + c_{add}$
    \State \Return $S$
  \end{algorithmic}
\end{algorithm}

Our work now deals with the application of this truncated Montgomery approach to the PMNS internal reduction.

\subsection{Truncated PMNS software implementation}

In the PMNS parameters, we choose $\phi = 2^{sw}$. This means that several machine words are used to store each polynomial coefficients.

Our idea is to combine Algorithm \ref{alg:redcoeff} with Algorithm \ref{alg:trunc_montred} coefficient wise. In Algorithm \ref{alg:trunc_montred}, the truncated approach is applied on a single modular multiplication of large integers. In this paper, it is applied to a polynomial coefficient computation, which is mainly an accumulation of the operand coefficient elementary multiplications. This leads to Algorithm \ref{alg:PMNS_trunc_montred}.

In the Algorithm \ref{alg:PMNS_trunc_montred} step 2, we apply the truncated operation in order to save the elementary multiplications which contribute to the least significant words of the polynomial coefficients. Thus we will replace the conventional multiplication by its truncated counter part, computing only the high part of each coefficient. Step 3, we also reconstruct the final result directly without computing the whole addition, but only the high part and the corresponding carry for each polynomial coefficients.

In detail, the initial operation step 2 of Algorithm \ref{alg:redcoeff} is

$$%\begin{array}{rl}
     T \gets  Q\times M \bmod E = \sum_{k=0}^{n-1} t_k X^k
%\end{array}
$$

with
$$%\begin{array}{rl}
     Q \gets  \sum_{k=0}^{n-1} q_k X^k
%\end{array}
$$
$$%\begin{array}{rl}
     M \gets \sum_{k=0}^{n-1} m_k X^k
%\end{array}
$$

$$t_k = \sum_{i+j=k} q_i m_j + \lambda\times\sum_{i+j=k+n} q_im_j$$

and, with a schoolbook type expression

$$q_im_j = \sum_{\substack{\ell=0\\\mu+\nu=\ell}}^{\ell=2s-2} q_i[\mu]m_j[\nu]2^{\ell w}$$

where $q_i[\mu]$ and $m_j[\nu]$ are respectively the words of weight $\mu$ of $q_i$ and $\nu$ of $m_j$.

In our new truncated approach, we only compute
$$\widehat T = \lfloor(Q\times M)/\phi\rfloor$$
that is 
$$\widehat T = \sum_{k=0}^{n-1} \widehat{t_k} X^k$$
with
$$\widehat{t_k} \gets \lfloor t_k / \phi\rfloor$$
which also is :

$$
\widehat{t_k} = \left\lfloor \frac{\underset{i+j=k}{\sum} q_i m_j + \lambda\times \underset{i+j=k+n}{\sum} q_im_j}{\phi} \right\rfloor
$$

We also set a polynomial of carries for each coefficients. This polynomial contains the carries which would have been propagated in the $C+T$ addition, that is for each coefficient, the carry of the additions coefficient-wise. This carry polynomial is denoted :

$$\mathcal C_{add} = \sum_{k=0}^{n-1} c_{add~k} X^k$$

with

$$c_{add~k} = \bigvee_{i=0}^{sw-1} c_k \langle i \rangle$$

\begin{algorithm}[ht]
  \caption{Montgomery-like PMNS Coefficient Reduction with truncated multiplication}
  \label{alg:PMNS_trunc_montred}
  \small
  \begin{algorithmic}[1]
    \Require  $\mathcal{B} = (p, \, n, \, \gamma, \, \rho, E)$ a PMNS, $C \in \mathbb{Z}_{n-1}[X]$ such that $\Vert C \Vert_{\infty} < 2 \rho$, $\phi \in \mathbb{N}\setminus\{0\}$, $M \isrep 0$
    and $M' \equiv -M^{-1}\ \text{mod}(E, \phi)$.
    \Ensure $S(\gamma) = C(\gamma)\phi^{-1} \bmod p$, with $S \in \mnsB$
    \State $Q\gets C\times M' \bmod (E,\phi)$\Comment{$q_k<\phi$}
    \State $\widehat{T} \gets \lfloor Q\times M /\phi \rfloor$
    \State $\mathcal C_{add} \gets \sum_{k=0}^{n-1} \left(\bigvee_{i=0}^{sw-1} c_k \langle i \rangle \right)X^k$
    \State $S \gets\lfloor C/\phi  \rfloor+ \widehat{T} + \mathcal C_{add}$
    \State \Return $S$
  \end{algorithmic}
\end{algorithm}

\subsection{Correctness of the truncated PMNS Montgomery-like internal reduction}

Theorem~\ref{th:correctness} provides the correctness of our PMNS Montgomery-like coefficient reduction with truncated multiplication. Theorem~\ref{th:qn} describes how to efficiently compute the correct $\widehat{T}$.
 
%\paragraph{Correctness of the truncated MontMul}
\begin{theorem}[Correctness of the Truncated PMNS Coefficient Reduction]
\label{th:correctness}
In the PMNS $\mathcal{B} = (p, \, n, \, \gamma, \, \rho, E)$, we set $C =\sum_{k=0}^{n-1} c_k X^k$ with $c_k < n\rho^2$, precomputed polynomial $M' \equiv (-M)^{-1} \bmod E$ with $M(\gamma) \equiv 0 \bmod p$.

If $\widehat{T}$ is correctly computed, then Algorithm~\ref{alg:PMNS_trunc_montred} correctly computes $S \equiv C \times \phi^{-1}\bmod N$ and $|c_k|<\rho, \forall k, 0\leq k<n$. 
\end{theorem}
\begin{proof}
In Algorithm~\ref{alg:redcoeff}, we remind that $T\gets Q\times M \bmod E$, and that the division by $\phi$ step 3 is exact. Thus, we know for sure that $C+T \bmod \phi \equiv 0$. This means that the $sw$ least significant bits of each polynomial coefficient of $C+T \bmod E$ are zeroes.
This makes easier the computation of the coefficients of the carry polynomial $\mathcal C_{add}$. If $\widehat{T}$ is correctly computed, then correct result of Algorithm~\ref{alg:PMNS_trunc_montred} depends on the the carry polynomial $\mathcal C_{add}$.    

If the least significant bit of $c_k$ is $0$, then the least significant bit of $t_k$ is also $0$, and no carry is propagated to the next position. This property holds until the first bit of $c_k$ equals~$1$.

Considering the first bit of $c_k$ which is 1, say the $i^{th}$ bit, then the $i^{th}$ bit of $t_k$ must also be 1 and the carry to be propagated to position $i+1$ must be 1 also, to ensure that the next $i^{th}$ bit of $c_k + t_k$ is 0. For now on, to ensure that the remaining most significant bits of $c_k + t_k$ are 0s, only one of the $i+1^{th}$ bits of $c_k$ and $t_k$ must be 1. These conditions are summarized in the table below:

\begin{center}
    \begin{tabular}{p{1.8cm}||c||c||c|c||c||}
   bits    &$0\ldots i-1$&   $i$   &   \multicolumn{2}{c||}{$i+1$}&$i+2\ldots$  \\
    $c_k$ &$0\ldots0$&  1 &   0  &   1   &$x\ldots x$\\
    $t_k$  &$0\ldots0$&   1   &   1  &   0   &$x\ldots x$\\
    \hline
    generated carry $c_{add~k}$   &$0\ldots0$&   1    &   \multicolumn{2}{c||}1&$1\ldots1$    \\
\end{tabular}
\end{center}

Thus, whenever a carry is found at an $i^{th}$ bit, then another carry is generated at position $i+1^{th}$, and so on. Therefore, the last carry at position $sw-1$, to be added to the $sw$ word of $c_k+t_k$ is 1 whenever at least one of the $sw -1$ least significant bits is 1. This carry is output from the $s-1^{th}$ word.
\begin{equation}\label{eq:c_add}
    c_{add~k}= \bigvee_{i=0}^{sw-1} c_k \langle i \rangle
\end{equation}
\end{proof}

%\paragraph{Computation of $\widehat{qN}$ }
\begin{theorem}[Computation of $\widehat{T}$]\label{th:qn}
The correct computation of $\widehat{T}$ requires only the words of weight at least $(s-1)$ of the coefficients of $Q\times M$ .    
\end{theorem}
\begin{proof}
    
Since we know that $C+Q\times M  \equiv 0 \bmod \phi$, the computation $\widehat{T}$ can be simplified.

We consider this computation coefficient-wise. It is not necessary to compute all the least significant words of $t_k$ and sum them in order to know which carries to propagate and to compute $\widehat{t_k}$ correctly.

Let us denote $\widetilde{(t_k)}[i]$ the sum of the partial product words of weight $i$. More formally, the partial product of weight $s-1$ is a $\left(w + \lceil \log_2(\lambda n s) \rceil +1 \right)$-bit word and is computed as follows: %%% LSD Vérifier les bornes

\begin{equation*}
    \begin{split}
        \widetilde{(t_k)}[s-1] = &\sum_{i+j=k} \widetilde{(q_i m_j)}[s-1]\\
        &+ \lambda\times\sum_{i+j=k+n} \widetilde{(q_im_j)}[s-1]\\
         &+ \widetilde{(t_k)}[s-2]>>w
    \end{split}
\end{equation*}

with

$$\widetilde{(q_im_j)} = \sum_{\substack{\ell=0\\\mu+\nu=\ell}}^{\ell=2s-2} q_i[\mu]m_j[\nu]2^{\ell w}$$

and

\begin{equation}%{clc}
\label{eq:qn}%{clc}
	\begin{split}
    \widetilde{(q_im_j)}[s-1]  \gets & \! \sum^{s-1}_{\substack{\mu=0,\nu=0\\ \mu+\nu=s-1}} \!\!\!\!\! \mulwlo(q_i[\mu],m_j[\nu])  \\
     &+ \!\!\!\!\!\! \sum^{s-2}_{\substack{\mu=0,\nu=0\\ \mu+\nu=s-2}} \!\!\!\!\! \mulwhi(q_i[\mu],m_j[\nu])
	\end{split}
\end{equation}
where $\mulwhi$ and $\mulwlo$ compute respectively the $w$ higher and lower bits of two $w$-bit operands.
Therefore, the $w-1^{th}$ word of $c_k+t_k$ (which is known to be 0) is:
$$(c_k+t_k)[s-1]=(c_k[s-1]+ \widetilde{(t_k)}[s-1] + c_{add~k} ) \bmod 2^w$$
and then
\begin{equation}\label{eq:TqN}
    (c_k[s-1]+ \widetilde{(t_k)}[s-1] + c_{add~k} ) \equiv 0 \bmod 2^w
\end{equation}

\textbf{If $\mathbf{c_{add~k}=0}$:} This means that all $t_k[\ell]=0$  for $\ell\leq s-1$ and so for $t_k[\ell]$, because
$$q_im_j= \sum_{\substack{\ell=0\\\mu+\nu=\ell}}^{\ell=2s-2} q_i[\mu]m_j[\nu]2^{\ell w}$$
in Alg.~\ref{alg:PMNS_trunc_montred}. This is possible if and only if $q_i[\mu]=0$ for $0\leq \mu < s$. Therefore in equation (\ref{eq:qn}), only the words of weight greater than or equal to $s-1$ are needed.

\textbf{If $\mathbf{c_{add~k}=1}$:} The only way the equation (\ref{eq:TqN}) can be verified is if the binary vector $c_k[s-1]+ \widetilde{(t_k)}[s-1]$ has all its $w$ least significant bits set to 1. 

Let us denote $\UP_k$ the accumulation of the elementary products composing the result $\widetilde{t_k}$, thus including the bits of weight $>w$ of each $s-1^{th}$ word of the coefficient of each $c_k$ above the $w$ first least significant bits.
Using Eq. (\ref{eq:qn}), equation (\ref{eq:TqN}) can be written as follows:
\begin{equation}
     (\UP_k + \widetilde{(t_k)}[s-2] >> w+ c_{add_k} ) \equiv 0 \bmod 2^{w}
\end{equation}
where 
\[%{clc}
%\label{eq:qn}%{clc}
	\begin{split}
    \UP_k  \gets & c_k[s-1] + \!\!\!\sum_{i+j=k} \widetilde{(q_i m_j)}[s-1]\\
        &+ \lambda\times\sum_{i+j=k+n} \widetilde{(q_im_j)}[s-1]\\
	\end{split}
\]

%%%%%%%%%%%%%%%%
\begin{figure}[ht]
    \centering

\scalebox{0.8}{\input{rect12new.pspdftex}}
    \caption{Detail of $\UP_k + (\widetilde{(t_k)}[s-2] >> w)$}
    \label{fig:}
\end{figure}
Therefore, $\widetilde{(t_k)}[s-2]>>w$ can be computed with the least significant bits of $\UP_k$ and its sum with $\UP_k$ does not generate any carry out of the $w^{th}$ bit. Thus $\widetilde{(t_k)}[s-1]>>w$ is computed only with the most significant bits of $\UP_k$ and the carry $c_{add~k}$ which is propagated to the $(w+1)^{th}$ bit of $\UP_k$. Figure \ref{fig:} shows the shape of $UP_k$.
 \medskip
 
 \begin{table*}[t]
    \centering
    \resizebox{\textwidth}{!}{
    \begin{tabular}{|l|c|c||r|r|r||r|r|r||r|r|r|}
    \hline
    & \multirow{2}{1em}{L} & \multirow{2}{1em}{T} & \multicolumn{3}{c||}{\textbf{2048}} & \multicolumn{3}{c||}{\textbf{4096}} & \multicolumn{3}{c|}{\textbf{8192}}\\
    & & & Full & Trunc & \cite{MeloniPV25} & Full & Trunc & \cite{MeloniPV25} & Full & Trunc & \cite{MeloniPV25}\\
    \hline
    \hline
    madd52lo & 4 & 0.5 & 651 & 546 & 672 & 2790 & 2340 & 1680 & 10332 & 8568 & 5424 \\
    \hline
    madd52hi & 4 & 0.5 & 609 & 539 & 576 & 2610 & 2310 & 1296 & 9828 & 8505 & 3888 \\
    \hline
    vpaddq & 1 & 0.5 & 81 & 76 & 18 & 242 & 232 & 447 & 627 & 606 & 1594 \\
    \hline
    vmov & 10 & 0.5 & 471 & 392 & 945 & 2529 & 1944 & 3894 & 9954 & 8387 & 11479 \\
    \hline
    permutex2var & 3 & 1 & 54 & 54 & 104 & 192 & 192 & 124 & 504 & 504 & 1158 \\
    \hline
    valign & 3 & 1 & 24 & 24 & 236 & 84 & 84 & 802 & 192 & 192 & 2472 \\
    \hline
    vextract & 3 & 1 & 48 & 48 & 30 & 132 & 120 & 186 & 504 & 505 & 423 \\
    \hline
    vpbroadcast & 3 & 1 & 126 & 127 & 60 & 270 & 271 & 438 & 504 & 505 & 423 \\
    \hline
    vpsrl/vpsll & 1 & 1 & 93 & 93 & 18 & 266 & 266 & 119 & 659 & 590 & 88 \\
    \hline
    \hline
    \multicolumn{3}{|c||}{Estimated cycles} & 1283 & 1155 & 1586 & 5062 & 4378 & 5360 & 17502 & 15096 & 15510 \\
    \hline
    \end{tabular} }
    \caption{Number of instructions in produced assembly files}
    \label{tab:instructions_in_code}
\end{table*}

As a conclusion $ \widehat{(t_k)}[0]$ is computed only with the partial results of $\widetilde{(q_im_j)}$ of weight at least equal $s-1$: 

\begin{equation}
\label{eq:qnhat}%{clc}
    \begin{split}
        \widehat{(t_k)}[0] = &\sum_{i+j=k} \widetilde{(q_i m_j)}[s]\\
        &+ \lambda\times\sum_{i+j=k+n} \widetilde{(q_im_j)}[s]\\
         &+ \widetilde{(t_k)}[s-1]>>w + c_{add~k}
    \end{split}
\end{equation}

\end{proof}

Thus, to compute these coefficients, likewise in the Algorithm \ref{alg:trunc_montred}, we remark that it is not necessary to compute all the elementary results of the $t_k[\ell]$. Knowing that the division by $\phi$ of $C+T$ is exact, we only compute the most significant part of the coefficients of these polynomials.

\begin{remark}
In the software implementation context, the bits of weight greater or equal to $w$ of $\widetilde{(t_k)}[s-1]$ have to be saved. When using the \texttt{AVX512} instruction extension set and the elementary multiplications of the fused multipliers-adder of 52-bits stored in 64-bit words, this is naturally done. We will detail this in the next sections dealing with the SIMD software implementation.
\end{remark}

\section{Complexity}

\subsection{Theoretical number of multiplications}

In table \ref{tab:count_madd} we sum up the theoretical number of elementary multiplications required to perform a complete multiplication in the PMNS.
Every polynomial multiplication makes use of the quadratic schoolbook algorithm (coefficient-wise and word-wise).
The gain of the truncated approach is then made explicit.
We include the number of multiplications announced in \cite{MeloniPV25} as well.

  \begin{table}
    \centering
    \resizebox{\textwidth}{!}{
      \begin{tabular}{|l||l|l||l|}
        \hline
        & \multicolumn{2}{c||}{\textbf{Schoolbook}} & \textbf{Toeplitz + CIOS\cite{MeloniPV25}}\\
        \hline
        \hline
        $ C \leftarrow A \times B \mod E$  &  \multicolumn{2}{c||}{$n^2 s^2$} & $n^{(2- \log_n (2))} s^2$ \\
        \hline
        $Q \leftarrow C \times M \mod (E, \phi)$ & \multicolumn{2}{c||}{$\frac{1}{2} n^2 s^2$} & $n^{(2- \log_n (2))} s$\\
        \hline
        \multirow{2}{10em}{$T \leftarrow Q \times M \mod E$} & \textit{Full} & $ n^2 s^2$ & \multirow{2}{5em}{$n^{(2- \log_n (2))} s^2$}\\
        & \textit{Trunc} & $n^2 s^2 - \frac{1}{2} n^2 {(s-1)}^2$ &\\
        \hline
        \hline
        \multirow{2}{8em}{Total} & \textit{Full} & $\frac{5}{2} n^2 s^2$ & \multirow{2}{13em}{$2n^{(2- \log_n (2))} s^2 + n^{(2- \log_n (2))}$}\\
        & \textit{Trunc} & $\frac{5}{2} n^2 s^2 - \frac{1}{2} n^2 (s-1)^2 $ & \\
        \hline
      \end{tabular}
    }
    \caption{Number of elementary multiplications in a PMNS multiplication }\label{tab:count_madd}
  \end{table}

\begin{remark}
As we show in table \ref{tab:instructions_in_code}, the number of \texttt{madd52lo} and \texttt{madd52hi} instructions effectively produced in the executable
from \cite{MeloniPV25} is much higher than the theoretical amount from table \ref{tab:count_madd}. We provide the explanation of this in section \ref{sec:res}.
\end{remark}

\subsection{Vectorized implementation}
The main advantage of PMNS relies on its abilities of paralleling computations.
Our code is fully vectorized using \texttt{AVX512} and \texttt{IFMA52} which implies $w=52$.
In these instructions set extensions are available multiplication-accumulations instructions called \texttt{\_mm512\_madd52lo\_epu64} and \texttt{\_mm512\_madd52hi\_epu64} referred to in this paper as \texttt{madd52lo} and \texttt{madd52hi} respectively.
These instructions are at the core of our multiplication algorithm.
Compared to \texttt{mul\_w\_lo} and \texttt{mul\_w\_hi}, they compute multiplication between the two last operands and add the 52
least (resp. most) significant bits of the result to the first one, all of this for a batch of 8 values.
Therefore, it is necessary to call \texttt{madd52lo} as well as \texttt{madd52hi} in order to perform one elementary multiplication over such a batch.

Let $\Bar{n} = \lceil n/8 \rceil$, we realize multiplications on batches, each of 8 values. As consequences, the number of calls to \texttt{madd52lo} and \texttt{madd52hi} combined for a full (resp. truncated) multiplication
in the PMNS becomes $5n \bar{n} s^2$ (resp. $5n \bar{n} s^2 - n \bar{n}(s-1)^2$).

\subsection{Memory Cost}

We compared the memory cost of our choices of parameters with those presented in \cite{MeloniPV25} and summed up results in table \ref{tab:mem_cost}.
We observe that a set of parameters specifically chosen to apply the efficient Toeplitz multiplication leads to an increased memory cost.

\begin{table}[h]
    \centering
    \begin{tabular}{|c|r|r||r|r||r|r|}
    \hline
    & \multicolumn{2}{c||}{\textbf{2048}} & \multicolumn{2}{c||}{\textbf{4096}} & \multicolumn{2}{c|}{\textbf{8192}} \\
    & $n$ & $s$ & $n$ & $s$ & $n$ & $s$\\
    \hline
    \hline
    \textbf{This work} & 7 & 6 & 15 &6 & 21 & 8\\
    \hline
    Size in bits & \multicolumn{2}{c||}{2184} & \multicolumn{2}{c||}{4680} & \multicolumn{2}{c|}{8736} \\
    \hline
    \hline
    \textbf{\cite{MeloniPV25}} & 16 & 3 & 32 & 3 & 64 & 3 \\
    \hline
    Size in bits & \multicolumn{2}{c||}{2496} & \multicolumn{2}{c||}{4992} & \multicolumn{2}{c|}{9984} \\
    \hline
    \hline
    Ratios & \multicolumn{2}{c||}{87.5\%} & \multicolumn{2}{c||}{93.75\%} & \multicolumn{2}{c|}{87.5\%}\\
    \hline
    \end{tabular}
    \caption{Memory cost for each set of parameters}
    \label{tab:mem_cost}
\end{table}

%%%%%%%%%%%%%%%%%%%%%%%%%%%%%%%%%%%%%%%%%%%%%%%%%%
\section{Performances}

\subsection{Experimentation procedure}
\label{sec:perfProc}

We compiled our executable with the following parameters:
"\texttt{-O3 -Wall -Wno-missing-braces -march=native -funroll-all-loops -fno-lto -lgmp -lcrypto}"

For all benchmarks, the \textit{Turbo-Boost}\textregistered~is deactivated.
We proceeded to measurements on two computers: one with an Intel Xeon w7-3555, one with an AMD Ryzen 9 7940HS.
We used the same procedure from \cite{MeloniPV25} and \cite{DidierEGR24} for performance measurement:
\begin {itemize}
	\item Run the code to be tested about 1300 times in order to "heat" the cache memory;
	\item 1200 random data sets are generated, and for each data set
      the minimum of the execution clock cycle numbers over a batch of 1800 runs is recorded;
	\item We take the average of all these minimums as the final measure.
\end{itemize}

All our source code is available at :
%\texttt{\href{https://gitlab.com/truncated_pmns/truncated_pmns}{https://gitlab.com/truncated\_pmns/truncated\_pmns}}
\url{https://gitlab.com/truncated\_pmns/truncated\_pmns}

All PMNS we benchmarked are in fact AMNS such that $E = X^n - 2$.
In order to compare our results with the state of the art \cite{MeloniPV25},
we run its source code which is available at %\texttt{\href{https://github.com/francoispalma/PMNS/tree/master/multiprecision_pmns}{https://github.com/francoispalma/PMNS/}} on our machines as well.
\url{https://github.com/francoispalma/PMNS/tree/master/multiprecision_pmns}

\subsection{Results}
\label{sec:res}

In Table \ref{tab:perfs}, we present our results and compare them to the results presented in \cite{MeloniPV25} which is, to the best of our knowledge, the State of the Art.
For multiple prime sizes, we give the number of cycles required for a complete multiplication, for a truncated one and the ratio between both.
We give the results obtained by running the source code of \cite{MeloniPV25} on our machines as well to compare them with ours.
We always remind the parameters used in our PMNS, $n$ is the maximum number of coefficients of a polynomial in the PMNS and
$s$ is the number of words of $52$ bits used to represent one coefficient of our polynomials.
It is worth noticing for $n=16$ and $s=3$, that is for a 2048 bit prime modulus, the source code used in \cite{MeloniPV25} implements a schoolbook approach for the polynomial multiplications.

\begin{table}[h]
  \centering
  \resizebox{\textwidth}{!}{
  \begin{tabular}{| c || c | c | c || c | c | c |}
    \hline
    CPU & \multicolumn{3}{c||}{\textbf{AMD Ryzen 9}} & \multicolumn{3}{c|}{\textbf{Intel Xeon}}\\
    \hline
    \hline
    Sizes & \textbf{2048} & \textbf{4096} & \textbf{8192} & \textbf{2048} & \textbf{4096} & \textbf{8192}\\
    \hline
    \hline
    \multicolumn{7}{|c|}{\textbf{This Work Schoolbook}}\\
    \hline
    $\mathbf{(n, s)}$ & $\mathbf{(7, 6)}$ & $\mathbf{(15, 6)}$ & $\mathbf{(21, 8)}$ & $\mathbf{(7, 6)}$ & $\mathbf{(15, 6)}$ & $\mathbf{(21, 8)}$\\
    \hline
    Full & 1407 & 6081 & 23671 & 997 & 3552 & 15817\\
    \hline
    Truncated & 1230 & 5268 & 19763 & 947 & 3046 & 13501\\
    \hline
    Trunc/Full & 87.4 \% & 86.3 \% & 83.5 \% & 95 \% & 85.7 \% & 85.4 \%\\
    \hline
    \hline
    \multicolumn{7}{|c|}{\textbf{Toeplitz + Montgomery CIOS\cite{MeloniPV25}}}\\
    \hline
    $\mathbf{n}$ & \textbf{24} & \textbf{48} & \textbf{120} & \textbf{24} & \textbf{48} & \textbf{120}\\
    \hline
    $\mathbf{s = 2}$ & 1920 & 6621 & 31751 & 1538 & 5212 & 24714\\
    \hline
    \hline
    $\mathbf{n}$ & \textbf{16} & \textbf{32} & \textbf{64} & \textbf{16} & \textbf{32} & \textbf{64}\\
    \hline
    $\mathbf{s = 3}$ & 1627 & 6151 & 20885 & 1197 & 5190 & 16365\\
    \hline
    \hline
    $\mathbf{n}$ & \diagbox{} & \textbf{24} & \textbf{48} & \diagbox{} & \textbf{24} & \textbf{48}\\
    \hline
    $\mathbf{s = 4}$ & \diagbox{} & 7029 & 23773 &  \diagbox{} &   5649 &  19094\\
    \hline
    Trunc/Best & 75.6 \% & 85.6 \% & 94.6 \% & 79.1 \% & 58.7 \% & 82.5 \%\\
    \hline
  \end{tabular}}

  \caption{Number of cycles of a multiplication in the PMNS}\label{tab:perfs}
\end{table}

The results in Table \ref{tab:perfs} show that our conventional non truncated schoolbook approach performs better than the Toeplitz PMNS multiplication from \cite{MeloniPV25}, for sizes of 2048 and 4096 bits. This is mainly explained by the more optimal PMNS parameters we use in this work. For these sizes, the Toeplitz approach requires $n$ being a power of 2  in order to build the recursive calls and split the operand. The efficient Toeplitz multiplication algorithm becomes better from 8192 bits, although, on the Xeon platform, it performs worse than our regular schoolbook implementation.
With our new truncated Montgomery-like approach, we observe significant speed-ups allowing this truncated schoolbook to outperform with the Toeplitz matrices multiplication, even in 8192 bits.

In order to explain why the Toeplitz efficient multiplication is worse than our regular schoolbook implementations, we provide some of the most important and costly instruction counts in the assembly file produced by the compiler from the source C codes (the one from \cite{MeloniPV25} and ours, full and truncated).
Table \ref{tab:instructions_in_code} shows the number of call to each instruction as well as an estimated cost in the hypothesis of a perfect pipeline, considering latency (L) and throughput (T) for Sapphire Rapids processors, based on Intel's documentation available at \cite{docIntel}. These counts shows that the instructions other than \texttt{madd52lo} and \texttt{madd52hi}, especially memory moves, are too costly to be neglected at least up to 4096 bits. In the implementations provided by \cite{MeloniPV25}, these costs, combined with the less favorable PMNS parameters, overshadow the benefits of the sub-quadratic Toeplitz multiplication approach.

\section{Conclusion}

In this article, we present a truncated approach of the Montgomery-like internal reduction in PMNS which is  similar to the improvement proposed by \cite{DidierEGR24} in the context of the Montgomery reduction for large multi-precision integers. We compare this new approach with the classical algorithm and we present the complexity of both resulting algorithms in terms of the number of elementary multiplications. This shows an improvement about 20 \% of the asymptotic cost.

We provide a software implementation of the classic full Montgomery-like approach using \texttt{AVX512/IFMA52} instruction set extensions, for modulus sizes 2048, 4096 and 8192 bits. This software implementation is competitive in comparison with the state-of-the-art, taking advantage of better parameter with a smaller memory cost. We then express the number of calls to \texttt{madd52lo} and \texttt{madd52hi} in both codes, which shows the same improvement. The benchmark of these software implementations present equivalent computation timings or speed-ups up to 20 \% compared to the s-o-a, according the modulus sizes.
Then, we experiment our truncated approach implemented with the same features. This new algorithm allows a speed-up up to 15\% compared to the Montgomery-like full algorithm taken as reference.

\bibliographystyle{alpha}
\bibliography{biblio}

\end{document}